\documentclass{amsart}
\usepackage{amsmath}
\usepackage{latexsym}
\usepackage{amssymb}
\usepackage{amscd}
\usepackage{amsfonts}
\usepackage{mathrsfs}

\newcommand{\oten}{\mathop{\overline{\otimes}}}

\newtheorem{theorem}{Theorem}

\newtheorem{lemma}{Lemma}
\newtheorem{cond}{Condition}

\theoremstyle{definition}

\newtheorem{remark}{Remark}

\begin{document}

\title{Discrete Approximation of Quantum Stochastic Models}
\thanks{This work was supported by the ARO under grant W911NF-06-1-0378.}

\author{Luc Bouten}
\address{Physical Measurement and Control 266-33, California Institute
of Technology, Pasadena, CA 91125, USA}
\email{bouten@its.caltech.edu}

\author{Ramon van Handel}
\address{Department of Operations Research and Financial Engineering,
Princeton University, Princeton, NJ 08544, USA}
\email{rvan@princeton.edu}

\begin{abstract}
We develop a general technique for proving convergence of repeated quantum 
interactions to the solution of a quantum stochastic differential 
equation.  The wide applicability of the method is illustrated in a 
variety of examples.  Our main theorem, which is based on the Trotter-Kato 
theorem, is not restricted to a specific noise model and does not require 
boundedness of the limit coefficients.
\end{abstract}

\maketitle

\section{Introduction}

It has been well established that the quantum stochastic equations
introduced by Hudson and Parthasarathy \cite{HuP84} provide an essential
tool in the theoretical description of physical systems, especially those
arising in quantum optics. The time evolution in these models is given by
a unitary cocycle that solves a Hudson-Parthasarathy quantum stochastic
differential equation (QSDE). These unitaries define a flow, which is a
quantum Markov process in the sense of \cite{AFL82}, that represents the
Heisenberg time evolution of the observables of the physical system.  
Several authors have studied how quantum stochastic models can be obtained
as a limit of fundamental models in quantum field theory
\cite{AFLu90,Gou05,DDR06}. This provides a sound justification for using
quantum stochastic models to describe physical systems arising, e.g., in
quantum optics.

In contrast to the limit theorems for field theoretic models, which are
non-Markovian and have continuous time parameter, it is natural to ask
whether QSDEs can be obtained as a limit of discrete time quantum Markov
chains.  Classical counterparts of such results are ubiquitous in
probability theory and there is a variety of motivations (to be discussed
further below) to study such limits. The first results on this topic date
back to the work of Lindsay and Parthasarathy \cite{Par88,LiP88}. In
\cite{Par88} it is shown that a particular class of repeated interaction
models, where a physical system is coupled to a spin chain, converge in a
very weak sense (in matrix elements) to the solution of a QSDE. A
significant step forward was taken in \cite{LiP88} where the authors embed
a chain of finite dimensional noise systems in the algebra of bounded
operators on the Fock space and show strong convergence of the discrete
flow to the flow obtained from a QSDE. Much later Attal and Pautrat
\cite{AtP06} obtained similar results (in the special case of a spin
chain) by showing that the discrete unitaries (rather than the flows)
converge strongly to the solution of a QSDE.

Independently from the previous work on the convergence of discrete 
chains, Holevo \cite{Hol92,Hol94} has studied a very similar problem in 
his work on time-ordered exponentials in quantum stochastic calculus.
The essence of Holevo's approach is to define time-ordered stochastic 
exponentials as the limit of discrete interaction models, where the role 
of the discrete noise is played by the increments of the field operators 
(the discrete noise is thus infinite dimensional in this setting, in 
contrast to the finite dimensional models considered by
Lindsay and Parthasarathy).  Despite the rather different motivation, the 
results of Holevo are strikingly similar to those obtained in the study of 
limits of discrete interaction models, as has been pointed out by
Gough \cite{Gou04}.  

None of these results, however, are capable of dealing with the physically
important case of limit QSDEs with unbounded initial coefficients (a
typical setting in, e.g., quantum optics).  This restriction is inherent
to the techniques used to prove these results, which rely on Dyson series
expansions and require boundedness of the coefficients.  Moreover, each of
the above results has been proved separately in its own setting, while the
similarity between these results strongly suggests that they should be
unified within a common framework.

The purpose of this paper is to introduce a general technique for proving
convergence of a sequence of discrete quantum Markov chains to the
solution of a QSDE.  Our approach does not rely on a Dyson series
expansion, but instead employs a form of the Trotter-Kato theorem.  This
allows us to deal with unbounded coefficients in a natural and transparent
manner.  In the simpler case where the limit coefficients and/or discrete
noises are bounded, we obtain many of the previous results as special
cases of our general theorems.  Moreover, the specific functional form of
the limiting coefficients, obtained in \cite{AtP06,Hol92,Gou04} by
identifying a power series obtained from the Dyson expansion, is
effectively demystified: we will see that it it is an immediate
consequence of the scaling of the noise operators.

Our motivation for this work is twofold.  First, the convergence of
discrete quantum Markov chains to continuous ones is a fundamental problem
in quantum probability theory.  In classical probability theory such
problems have been investigated for many decades, and the theory has
culminated in the well known work of Stroock and Varadhan \cite[section
11.2]{SV79}, Ethier and Kurtz \cite{EK86} and Kushner \cite{Kus84}, among
others.  A similar systematic investigation was hitherto lacking the the
quantum probability literature.  This paper presents one attempt to unify
and extend the existing results in this direction.

Second, we are motivated by practical problems in which one is
specifically interested in the convergence of discrete to continuous
models.  For example, certain laboratory experiments, e.g., atomic beam
experiments with a large flux of atoms, can be approximately modelled by
quantum stochastic equations in the limit of large flux.  Another
application of independent interest is the development of numerical
methods for quantum stochastic models.  To perform tractable numerical
simulations one is often forced to discretize, particularly in dynamical
optimization problems which appear in the emerging field of quantum
engineering \cite{BHJ06}, and convergence of the discretized
approximations is a challenging topic. We were motivated in particular by
the problem of discretizing \textit{linear} quantum systems, which play a
special role in linear systems theory \cite{JNP07,JNP07b,Mab08}, but are
not covered by previous results as both the noise and the initial system
are necessarily unbounded.

As compared to previous results, our approach is closest to the original
method of Lindsay and Parathasarathy \cite{LiP88}.  The simple uniform
convergence result \cite[proposition 3.3]{LiP88} is replaced in our
setting by a variant of the Trotter-Kato theorem \cite{EK86}, which allows
us to deal with the analytic complications inherent to the case of
unbounded coefficients.  We also work directly with the unitary evolution
rather than with the flow.  The Trotter-Kato theorem allows us to obtain
convergence by studying generators, and exploits in a fundamental way the
Markov property of both the approximate and the limit evolutions.  

Techniques to obtain convergence for QSDEs by studying generators were
introduced by Fagnola \cite{Fag93} and Chebotarev \cite{Che03} (using
resolvents) and by Lindsay and Wills \cite{LiW06,LiW07} (using the
Trotter-Kato theorem).  We have previously applied related techniques to
obtain general results on singular perturbation problems for quantum
stochastic models \cite{BoS08,BHS08}.  The application of techniques of
this type to obtain the convergence of discrete quantum models is new.

The remainder of this paper is organized as follows.  In section
\ref{sec:main} we introduce the class of discrete interaction models and
limit models which will be of interest throughout the paper, and we state
our main results.  The main theorem is a generalization of the
Trotter-Kato theorem to quantum stochastic models, and is generally
applicable.  We also introduce a more restricted family of discrete models
for which the conditions of this result can be verified explicitly.
Section \ref{sec:ex} develops a variety of known and new examples using
our results.  Finally, section \ref{sec:proofs} is devoted to the proofs
of our results.

\section{Main Results}
\label{sec:main}

In the following subsections, we first define the class of models
that we will consider and introduce the necessary assumptions.
This is followed by the statement of our main results.  The proofs of
our main results are contained in section \ref{sec:proofs} below.

\subsection{The limit model}

Throughout this paper we let $\mathcal{H}$, the initial space, be a
separable (complex) Hilbert space. We denote by
$\mathcal{F}=\Gamma_s(L^2(\mathbb{R}_+;\mathbb{C}^n))$ the symmetric Fock
space with multiplicity $n\in\mathbb{N}$ (i.e., the one-particle space is
$\mathbb{C}^n\otimes L^2(\mathbb{R}_+) \cong
L^2(\mathbb{R}_+;\mathbb{C}^n)$), and by $e(f)$, $f\in
L^2(\mathbb{R}_+;\mathbb{C}^n)$ the exponential vectors in $\mathcal{F}$.
The annhiliation, creation and gauge processes on $\mathcal{F}$, as well
as their ampliations to $\mathcal{H}\otimes\mathcal{F}$, will be denoted
as $A_t^i$, $A_t^{i\dag}$ and $\Lambda_t^{ij}$, respectively (the channel
indices are relative to the canonical basis of $\mathbb{C}^n$).
Moreover, we will fix once and for all a dense domain
$\mathcal{D}\subset\mathcal{H}$ and a dense domain of exponential vectors
$\mathcal{E}=\mathrm{span}\{e(f):f\in \mathfrak{S}\}\subset \mathcal{F}$,
where $\mathfrak{S}\subset L^2(\mathbb{R}_+;\mathbb{C}^n)\cap
L^\infty_{\rm loc}(\mathbb{R}_+;\mathbb{C}^n)$ is an admissible subspace
in the sense of Hudson-Parthasarathy \cite{HuP84} which is presumed to
contain at least all simple functions.  An introduction to these concepts 
(using a similar notation to the one used here) can be found in 
\cite{Bar03}, and we refer to \cite{HuP84,Par92} for a detailed 
description of quantum stochastic calculus.

Consider a quantum stochastic differential equation of the form
\begin{equation}
\label{eq:hplim}
        dU_t =
        U_t\left\{
        \sum_{i,j=1}^n (N_{ij}-\delta_{ij})\, d\Lambda^{ij}_t
        + \sum_{i=1}^n M_i\, dA^{i\dag}_{t}
        + \sum_{i=1}^n L_i\, dA^i_t  + K\, dt
        \right\},
\end{equation}
where $U_0 =I$ and the quantum stochastic integrals are defined
relative to the domain $\mathcal{D}\oten\mathcal{E}$ (for simplicity, we 
will use the same notation for operators on $\mathcal{H}$ or on 
$\mathcal{F}$ and for their ampliations to $\mathcal{H}\otimes\mathcal{F}$).
Under certain conditions to be introduced below, the solution of this 
equation describes the time evolution of quantum stochastic models such as
those used in quantum optics.  The purpose of this paper is to prove that 
the solution of this equation may be approximated by appropriately chosen
discrete interaction models.  The approximating models will be defined in 
the following subsection.

Denote by $\theta_t:L^2([t,\infty[\mbox{};\mathbb{C}^n)\to
L^2(\mathbb{R}_+;\mathbb{C}^n)$ the canonical shift $\theta_tf(s)=f(t+s)$,
and by $\Theta_t:\mathcal{F}_{[t}\to\mathcal{F}$ its second quantization
(here $\mathcal{F}\cong\mathcal{F}_{t]}\otimes\mathcal{F}_{[t}$ denotes
the usual continuous tensor product decomposition). Recall that an
adapted process $\{U_t:t\ge 0\}$ on $\mathcal{H}\otimes{\mathcal{F}}$ is
called a \textit{unitary cocycle} if $U_t$ is unitary for all
$t\ge 0$, $t\mapsto U_t$ is strongly continuous and
$U_{s+t}=U_s(I\otimes\Theta_s^*U_t\Theta_s)$, where
$I\otimes\Theta_s^*U_t\Theta_s$ is viewed as an operator on
$\mathcal{F}_{s]}\otimes(\mathcal{H}\otimes\mathcal{F}_{[s})\cong
\mathcal{H}\otimes\mathcal{F}$.  The following condition will always be 
presumed to be in force.

\begin{cond}
\label{cond:limit}
The operators $K$, $L_i$, $M_i$, and $N_{ij}$, defined on the domain
$\mathcal{D}$, are such that the Hudson-Parthasarathy equation 
(\ref{eq:hplim}) possesses a unique solution $\{U_t:t\ge 0\}$ which 
extends to a unitary cocycle on $\mathcal{H}\otimes\mathcal{F}$.
\end{cond}

\begin{remark}
\label{rem:hp}
When the coefficients $K$, $L_i$, $M_i$, $N_{ij}$ are bounded, it is well
known \cite{HuP84} that condition \ref{cond:limit} holds true if and only 
if the following algebraic relations are satisfied:
\begin{eqnarray*}
        &&K+K^* = -\sum_{i=1}^n L_iL_i^*,\qquad
        M_i = -\sum_{j=1}^n N_{ij}L_j^*, \\
        &&\sum_{j=1}^n N_{mj}N_{\ell j}^* =
        \sum_{j=1}^n N_{jm}^* N_{j\ell} =
        \delta_{m\ell}.
\end{eqnarray*}
Remarks on the verification of condition \ref{cond:limit} in the 
unbounded case can be found in \cite{BHS08}.
\end{remark}

\begin{remark}
\label{rem:lefteq}
We have chosen the \textit{left} Hudson-Parthasarathy equation
(\ref{eq:hplim}), rather than the more familiar
\textit{right} equation where the solution is placed to the right of
the coefficients.  This means that the Schr\"odinger evolution of a state
vector $\psi\in\mathcal{H}\otimes\mathcal{F}$ is given by $U_t^*\psi$,
etc.  The reason for this choice is that for equations with unbounded 
coefficients, it is generally much easier to prove the existence of a 
unique cocycle solution for the left equation than for the right equation 
(see, e.g., \cite{Fag90,LiW06}).  As we will ultimately prove convergence 
of discrete evolutions to $U_t^*$, there is no loss of generality in 
working with the more tractable left equations.  If we wish to begin with 
a well defined right equation (as is more natural when the coefficients 
are bounded), our results can be immediately applied to the 
Hudson-Parthasarathy equation for its adjoint.
\end{remark}

\subsection{The discrete approximations}

A discrete interaction model describes the repeated interaction of an
initial system with independent copies of an external noise source.
Given an initial Hilbert space $\mathcal{H}'$ and a noise Hilbert space
$\mathcal{K}'$, a single interaction is described by a unitary operator
$R'$ on $\mathcal{H}'\otimes\mathcal{K}'$.  To describe repeated 
interactions, we work on the Hilbert space $\mathcal{H}'\otimes
\bigotimes_{\mathbb{N}}\mathcal{K}'$ on which we define the natural
isomorphism
$$
	\Xi_k:(\mathcal{K}')^{\otimes (k-1)}\otimes
	\mathcal{H}'\otimes\bigotimes_{\mathbb{N}}\mathcal{K}'\to
	\mathcal{H}'\otimes\bigotimes_{\mathbb{N}}\mathcal{K}'
$$
as $\Xi_k(\psi_{k-1]}\otimes\xi\otimes\psi_{[k})=
\xi\otimes(\psi_{k-1]}\otimes\psi_{[k})$.
We now define recursively $R_0'=I$ and
$$
	R_k'(\xi\otimes \psi_{k-1]}\otimes\psi_k\otimes\psi_{[k+1}) =
	R_{k-1}'\Xi_k(\psi_{k-1]}\otimes 
	R'(\xi\otimes\psi_k)\otimes\psi_{[k+1})
$$
for $k\in\mathbb{N}$, where $\psi_{k-1]}\otimes\psi_k\otimes\psi_{[k+1}\in 
(\mathcal{K}')^{\otimes (k-1)}\otimes 
\mathcal{K}'\otimes\bigotimes_{\mathbb{N}}\mathcal{K}'\cong 
\bigotimes_{\mathbb{N}}\mathcal{K}'$ and $\xi\in\mathcal{H}'$. This is 
precisely the discrete counterpart of a unitary cocycle.  Note that the 
order of multiplication in the recursion for $R_k'$ matches the choice 
of the \textit{left} Hudson-Parthasarathy equation above (i.e., 
$R_k^{\prime *}\psi$ is the Schr{\"o}dinger evolution, see remark 
\ref{rem:lefteq}).

The purpose of this paper is to prove that the solution of the
Hudson-Parthasarathy equation (\ref{eq:hplim}) can be approximated by a
sequence of discrete interaction models with decreasing time step.  In
order to study this problem, we will embed our discrete interaction models
in the limit Hilbert space $\mathcal{H}\otimes\mathcal{F}$.  This
allows us to prove strong convergence of the embedded discrete cocycles to
the solution of equation (\ref{eq:hplim}).  The precise way in which the
embedding is done does not affect the proof of our main result; we
therefore proceed in a general fashion by defining a fixed but otherwise
arbitrary sequence of discrete interaction models which are already
embedded into the limit Hilbert space (this avoids, without any loss of
generality, the notational burden of introducing separate Hilbert spaces
and embedding maps for every discrete approximation).  As a special case 
of this general construction, we will introduce in section 
\ref{sec:explicit} below an interesting class of discrete 
interaction models for which the embedding is made explicit.

We proceed to introduce the embedded discrete interaction models.
For every $k\in\mathbb{N}$, we will define a discrete interaction model 
with time step\footnote{
	This choice is only made to keep our notation manageable,
	and is by no means a restriction of the method.
} $2^{-k}$ (the ultimate goal being to take the limit as $k\to\infty$).
By its continuous tensor product property, the Fock space is isomorphic
to $\mathcal{F}\cong\bigotimes_{\mathbb{N}}\mathcal{F}_{2^{-k}]}$, where
each component $\mathcal{F}_{2^{-k}]}$ represents a consecutive time 
slice of length $2^{-k}$.  Our discrete interaction models will be 
embedded as repeated interactions with consecutive time slices of the 
field.  Let us make this precise: the following notations will be used 
throughout the paper.  Let $\mathcal{H}^k\subset\mathcal{H}$ be the 
initial Hilbert space and let $\mathcal{K}^{k}\subset\mathcal{F}_{2^{-k}]}$
be the noise Hilbert space for the interaction model with time step 
$2^{-k}$.  We will write
$$
	\mathcal{F}^k=\bigotimes_{\mathbb{N}}\mathcal{K}^k
	\subset \bigotimes_{\mathbb{N}}\mathcal{F}_{2^{-k}]}
	\cong\mathcal{F}.
$$
We now introduce a unitary operator $R^k$ on $\mathcal{H}^k\otimes
\mathcal{K}^k$ which describes the interaction in a single time step, and 
we extend $R^k$ to $\mathcal{H}\otimes\mathcal{F}$ by setting 
$R^k\psi=\psi$ for $\psi\in(\mathcal{H}^k\otimes\mathcal{K}^k\otimes
\mathcal{F}_{[2^{-k}})^\perp$. We now define recursively $R_t^k = I$ 
for $t\in[0,2^{-k}[\mbox{}$ and
$$
	R_t^k = R_{(\ell-1) 2^{-k}}^k
	(I\otimes\Theta_{(\ell-1) 2^{-k}}^*R^k\Theta_{(\ell-1) 2^{-k}})
	\quad\mbox{for }t\in[\ell 2^{-k},(\ell+1)2^{-k}[\mbox{},\quad
	\ell\in\mathbb{N},
$$
where $I\otimes\Theta_{\ell 2^{-k}}^*R^k\Theta_{\ell 2^{-k}}$ is 
viewed as an operator on $\mathcal{F}_{\ell 2^{-k}]}\otimes
(\mathcal{H}\otimes\mathcal{F}_{[\ell 2^{-k}})\cong \mathcal{H}\otimes
\mathcal{F}$.  In words, the discrete evolution $R_t^k$ is an adapted 
unitary process which is piecewise constant on consecutive intervals of 
length $2^{-k}$, and an interaction between the initial system and the 
next field slice occurs at the beginning of every interval.  Our goal is 
to prove that $R_t^k$ converges as $k\to\infty$ to the solution $U_t$ of 
equation (\ref{eq:hplim}).

\begin{remark}
This model coincides precisely with the discrete interaction model
defined above if we choose $\mathcal{H}'=\mathcal{H}^k$, 
$\mathcal{K}'=\mathcal{K}^k$, $R'=R^k$, and $R_\ell'=R^k_{\ell 2^{-n}}$.
\end{remark}

\subsection{A general limit theorem}

Before we proceed to the statement of our main result, we introduce 
certain families of semigroups which will play a central role in our
approach. 

\begin{lemma}
\label{lem:discgen}
For $k\in\mathbb{N}$ and $\psi,\varphi\in\mathcal{F}_{2^{-k}]}$, define
$R^{k;\psi\varphi}:\mathcal{H}\to\mathcal{H}$ such that
$$
	\langle u,R^{k;\psi\varphi}v\rangle =
	\|\psi\|^{-1}\|\varphi\|^{-1}\langle u\otimes\psi,
	R^k\,v\otimes\varphi\rangle\qquad
	\forall\,u,v\in\mathcal{H}.
$$
Then $R^{k;\psi\varphi}$ is a contraction on $\mathcal{H}$ and
$$
	\langle u,(R^{k;\psi\varphi})^{\lfloor t2^k\rfloor}v\rangle =
	\|\psi\|^{-N2^k}\|\varphi\|^{-N2^k}
	\langle u\otimes\psi^{\otimes N2^k},
	R_t^k\,v\otimes\varphi^{\otimes N2^k}\rangle
$$
for all $u,v\in\mathcal{H}$ and $t\in [0,N]$, $N\in\mathbb{N}$.
\end{lemma}

\begin{proof}
This follows directly from the unitarity of $R^k$ and the definition
of $R^k_t$.
\end{proof}

The following counterpart for the limit equation (\ref{eq:hplim})
is proved in \cite[lemma 1]{BHS08}.

\begin{lemma}
\label{lem:gens}
For $\alpha,\beta\in\mathbb{C}^n$, define $T_t^{\alpha\beta}:
\mathcal{H}\to\mathcal{H}$ such that
\begin{equation*}
        \langle u,T_t^{\alpha\beta}v\rangle = e^{-(|\alpha|^2+
                |\beta|^2)t/2}
        \langle u\otimes e(\alpha I_{[0,t]}),
        U_t\,v\otimes e(\beta I_{[0,t]})\rangle
        \qquad \forall\,u,v\in\mathcal{H},~t\ge 0.
\end{equation*}
Then $T_t^{\alpha\beta}$ is a strongly continuous contraction
semigroup on $\mathcal{H}$, and the generator
$\mathscr{L}^{\alpha\beta}$ of this semigroup satisfies
$\mathrm{Dom}(\mathscr{L}^{\alpha\beta})\supset\mathcal{D}$
such that for $u\in\mathcal{D}$
\begin{equation*}
        \mathscr{L}^{\alpha\beta}u =
        \left(\sum_{i,j=1}^n \alpha_i^* N_{ij} \beta_j +
        \sum_{i=1}^n \alpha^*_i M_i +
        \sum_{i=1}^n L_i \beta_i + K
        - \frac{|\alpha|^2+|\beta|^2}{2}
        \right)u.
\end{equation*}
\end{lemma}

The reason to focus on the semigroups associated to our models is that we
will seek conditions for convergence of the discrete approximations in
terms of the generators $\mathscr{L}^{\alpha\beta}$.  As the latter are
expressed directly in terms of the coefficients of the limit equation
(\ref{eq:hplim}), such conditions can typically be verified in a
straightforward manner and do not require us to work directly with the 
solution $U_t$ of that equation.

The following is the main result of this paper.  The theorem bears strong
resemblance to the Trotter-Kato theorem for contraction semigroups, and
the latter does indeed form the foundation of the proof.  The proof of the
theorem is given in section \ref{sec:proofs}.

\begin{theorem}
\label{thm:main}
Assume that condition \ref{cond:limit} holds, and let 
$\mathcal{D}^{\alpha\beta}\subset 
\mathrm{Dom}(\mathscr{L}^{\alpha\beta})$ be a core for
$\mathscr{L}^{\alpha\beta}$, $\alpha,\beta\in\mathbb{C}^n$.  Then the
following conditions are equivalent.
\begin{enumerate}
\item For all $\alpha,\beta\in\mathbb{C}^n$,
$u\in\mathcal{D}^{\alpha\beta}$ there exist $u^{k}\in\mathcal{H}$
and $\psi^{k},\varphi^{k}\in\mathcal{F}_{2^{-k}]}$ so that
$$
	u^{k}\xrightarrow{k\to\infty}u,\quad
	(\psi^{k})^{\otimes 2^k}\xrightarrow{k\to\infty}
	e(\alpha I_{[0,1]}),\quad
	(\varphi^{k})^{\otimes 2^k}\xrightarrow{k\to\infty}
	e(\beta I_{[0,1]}),
$$
and
$$
	2^k(R^{k;\psi^k\varphi^k}-I)u^k
	\xrightarrow{k\to\infty}\mathscr{L}^{\alpha\beta}u.
$$
\item
For every $T<\infty$ and $\psi\in\mathcal{H}\otimes\mathcal{F}$
$$
	\lim_{k\to\infty}\sup_{0\le t\le T}
	\|R_t^{k*}\psi-U_t^*\psi\|=0.
$$
\end{enumerate}
\end{theorem}

\begin{remark}
As was pointed out to us by a referee, this theorem could also be stated 
outside the framework of quantum stochastic differential equations.
Indeed, a careful reading of the proof reveals that one may dispose of 
condition \ref{cond:limit} entirely and assume only that $U_t$ is any 
unitary cocycle on $\mathcal{H}\otimes\mathcal{F}$.  The theorem is the 
most powerful, however, when the generators $\mathscr{L}^{\alpha\beta}$ 
admit an explicit expression in terms of the parameters of the limit 
model.  This is the case, in particular, when condition \ref{cond:limit} 
is satisfied and $\mathcal{D}$ is a core for all 
$\mathscr{L}^{\alpha\beta}$, $\alpha,\beta\in\mathbb{C}$.  We can then 
choose $\mathcal{D}^{\alpha\beta}\subset\mathcal{D}$, with the 
important consequence that this puts the explicit expression for 
$\mathscr{L}^{\alpha\beta}$ in lemma \ref{lem:gens} at our disposal.  This 
will be the case in all our examples.  Typically existence and uniqueness
proofs for the solution of (\ref{eq:hplim}) already imply that 
$\mathcal{D}$ is a core for $\mathscr{L}^{\alpha\beta}$, see, e.g.,
\cite{Fag90,LiW06} and \cite[remark 4]{BHS08} for further comments.
\end{remark}

\begin{remark}
The assumption in condition \ref{cond:limit} that $U_t$ is a 
\textit{unitary} cocycle can be weakened somewhat. In the absence of this 
assumption, one may still prove the implication 1$\Rightarrow$2 provided 
that strong convergence uniformly on compact intervals is replaced by weak 
convergence of $R_t^k$ to $U_t$ for every time $t$.  We have chosen 
to concentrate on the unitary case, as it is the relevant one for physical 
applications and admits a much stronger result.
\end{remark}

\subsection{A class of discrete interactions}
\label{sec:explicit}

Theorem \ref{thm:main} is a very general result which allows us to infer
the convergence of a sequence of discrete interaction models to the
solution of the Hudson-Parthasarathy equation \ref{eq:hplim}.  The
verification of the conditions of the theorem requires additional work,
however, and the form of the limit coefficients depends on the choice of
the discrete interactions $R^k$.  In this section we introduce a special
class of discrete interaction models (with $\mathcal{H}^k=\mathcal{H}$) 
for which the conditions of theorem \ref{thm:main} can be verified 
explicitly.  In particular, we obtain explicit expressions for the limit 
coefficients.  It should be noted that this class of discrete interaction 
models is physically natural, and we will encounter several concrete
examples in section \ref{sec:ex}.

Let $\mathcal{K}$ be a fixed Hilbert space and suppose that we are given a 
sequence of bounded operators $\pi^k:\mathcal{K}\to\mathcal{F}_{2^{-k}]}$ 
which are partially isometric in the following sense: 
$$
	\pi^{k*}\pi^k=I_{\mathcal{K}},\quad
	\pi^k\pi^{k*}=P_{\mathcal{K}^k}\quad
	\mbox{for all }k,\qquad
	\mbox{where }\mathcal{K}^k:=\mbox{ran}\,\pi^k.
$$ 
Here $P_{\mathcal{K}^k}$ is the orthogonal projection onto $\mathcal{K}^k$ 
and $I_{\mathcal{K}}$ is the identity on $\mathcal{K}$.
The same space $\mathcal{K}$ will play the role of the noise Hilbert space 
for every discrete interaction model $R^k$ ($k\in\mathbb{N}$) after being
isometrically embedded into the limit Hilbert space by the embedding maps 
$\pi^k$.  As we will see, the choice to work with a fixed noise space 
prior to embedding is convenient due to the fact that the limit 
coefficients can be expressed in terms of the matrix elements of certain 
operators on $\mathcal{H}\otimes\mathcal{K}$.

We now define the discrete models.  Let $F_1,\ldots,F_\ell$,
$G_1,\ldots,G_m$, and $H_1,\ldots,H_r$ ($\ell,m,r\in\mathbb{N}$) be
bounded self-adjoint operators on $\mathcal{H}$, and let
$\lambda_1,\ldots,\lambda_\ell$, $\mu_1,\ldots,\mu_m$ and
$\nu_1,\ldots,\nu_r$ be (not necessarily bounded)  self-adjoint operators
on $\mathcal{K}$.  Define
$$
	H^k = 2^k \sum_{j=1}^\ell F_j\otimes\lambda_j +
	2^{k/2}\sum_{j=1}^m G_j\otimes \mu_j + \sum_{j=1}^r 
	H_j\otimes\nu_j
$$
on $\mathcal{H}\oten(\bigcap_j\mathrm{Dom}(\lambda_j)\cap
\bigcap_j\mathrm{Dom}(\mu_j)\cap\bigcap_j\mathrm{Dom}(\nu_j)):= 
\mathcal{D}_0$.

\begin{cond}
\label{cond:expl}
The following are presumed to hold:
\begin{enumerate}
\item $H^k$ is essentially self-adjoint for every $k$.
\item $\check F:=F_1\otimes\lambda_1+\cdots+F_\ell\otimes\lambda_\ell$ is 
essentially self-adjoint on the domain $\mathcal{D}_0$.
\item There is a family of orthonormal vectors $\chi_0,\ldots,\chi_n\in
\mathcal{K}$ such that for all $\alpha\in\mathbb{C}^n$, the vector 
$\chi^k(\alpha):=\chi_0+2^{-k/2}\sum_{j=1}^n\alpha_j\chi_j$ satisfies
$(\pi^k\chi^k(\alpha))^{\otimes 2^k}\to e(\alpha I_{[0,1]})$.
\item $\chi_0\in \mathrm{Dom}(\nu_j)$ for all $j$.
\item $\chi_0\in \mathrm{Dom}(\mu_j)$ and 
$\langle\chi_0,\mu_j\chi_0\rangle=0$ for all $j$.
\item $\chi_0\in \mathrm{Dom}(\lambda_j)$ and $\lambda_j\chi_0=0$ for all $j$.
\end{enumerate}
\end{cond}

We subsequently identify $H^k$ and $\check F$ with their closures.
We now define the discrete interaction unitary 
$R^k:=\pi^ke^{iH^k2^{-k}}\pi^{k*}$ on $\mathcal{H}\otimes\mathcal{K}^k$, 
and extend to $\mathcal{H}\otimes\mathcal{F}$ as usual.

To state the convergence result for this class of discrete interaction 
models we must introduce the corresponding limit coefficients $N_{ij}$, 
$M_i$, $L_i$ and $K$ in equation (\ref{eq:hplim}).  This is what we turn 
to presently.  Define the bounded continuous functions 
$f,g:\mathbb{R}\to\mathbb{C}$ as
$$
	f(x) = \frac{e^{ix}-1}{x},\qquad\quad
	g(x) = \frac{e^{ix}-ix-1}{x^2}.
$$
We must also define the bounded operators $W_{ij},X_i^p,Y_i^p,Z^{pq}$ 
($i,j=1,\ldots,m$, $p,q=1,\ldots,n$) on $\mathcal{H}$ as follows: for 
every 
$u,v\in\mathcal{H}$,
\begin{equation*}
\begin{split}
	&\langle u,W_{ij}\,v\rangle =
	\langle u\otimes\mu_i\chi_0,g(
		\check F)\,v\otimes\mu_j\chi_0\rangle,
	\\
	&\langle u,X_i^p\,v\rangle =
	\langle u\otimes\chi_p,f(
		\check F)\,v\otimes\mu_i\chi_0\rangle,
	\\
	&\langle u,Y_i^p\,v\rangle =
	\langle u\otimes\mu_i\chi_0,f(
		\check F)\,v\otimes\chi_p\rangle,
	\\
	&\langle u,Z^{pq}\,v\rangle =
	\langle u\otimes\chi_p,e^{i\check F}\,v\otimes\chi_q\rangle.
\end{split}
\end{equation*}
It is evident that these operators are well defined provided that 
condition \ref{cond:expl} is assumed to hold.  We now define the limit 
coefficients as follows.
\begin{equation*}
\begin{split}
	&N_{pq} = Z^{pq},\qquad
	M_p = \sum_{i=1}^m X_i^p G_i, \qquad
	L_p = \sum_{i=1}^m G_i Y_i^p, \\
	&K = i\sum_{j=1}^r H_j\,\langle\chi_0,\nu_j\chi_0\rangle
	+ \sum_{i,j=1}^m G_i W_{ij} G_j.
\end{split}
\end{equation*}
The following is the main result of this section.

\begin{theorem}
\label{thm:expl}
Suppose that condition \ref{cond:expl} holds and that condition 
\ref{cond:limit} holds for the limit coefficients $N_{pq},M_p,L_p,K$ as 
defined above (with $\mathcal{D}=\mathcal{H}$).  Then
$$
	\lim_{k\to\infty}\sup_{0\le t\le T}
	\|R_t^{k*}\psi-U_t^*\psi\|=0\quad
	\mbox{for all }\psi\in\mathcal{H}\otimes\mathcal{F},~
	T<\infty.
$$
\end{theorem}

\begin{remark}
We have not sought to develop this result under the most general
conditions possible.  In particular, the boundedness of $F_j,G_j,H_j$ can
certainly be relaxed if appropriate domain assumptions are introduced, and
the proof of the present result is then readily extended. We have chosen
to restrict to the bounded case as the treatment of this case is
particularly transparent, and we have not found one single choice of
domain conditions in the unbounded case which covers all examples of
interest in that setting (particularly straightforward extensions can be
found when either $F_j=0$ for all $j$, or when only the $G_j$ are 
bounded). An illustrative example with unbounded coefficients is developed 
in section \ref{sec:ex} by appealing directly to theorem \ref{thm:main}.
\end{remark}

One might worry about the well-posedness of theorem \ref{thm:expl} as the 
limit coefficients $N_{pq},M_p,L_p,K$ depend on our choice for 
$\chi_0,\ldots,\chi_n$.  For completeness, we provide the following simple 
lemma whose proof can be found in section \ref{sec:proofs}.

\begin{lemma}
\label{lem:uniqueness}
Define $\chi^k(\alpha)$ as in condition \ref{cond:expl} (which we presume 
to be in force).  Suppose that 
$\tilde\chi_0,\ldots,\tilde\chi_n\in\mathcal{K}$ is another
orthonormal family such that for every $\alpha\in\mathbb{C}^n$
$$
	\tilde\chi^k(\alpha):=
	\tilde\chi_0+2^{-k/2}\sum_{j=1}^n\alpha_j\tilde\chi_j
	\quad\mbox{satisfies}\quad
	(\pi^k\tilde\chi^k(\alpha))^{\otimes 2^k}
	\xrightarrow{k\to\infty} e(\alpha I_{[0,1]}).
$$
Then there is a $\phi\in\mathbb{R}$ such that
$\tilde\chi_j=e^{i\phi}\chi_j$ for all $j$.
\end{lemma}

Finally, it should be noted that condition \ref{cond:limit} imposes 
stronger assumptions on the noise operators $\mu_j$ and $\lambda_j$ than 
is evident from condition \ref{cond:expl}.  The following lemma provides 
explicit assumptions which guarantee that condition \ref{cond:limit} holds 
for the coefficients $N_{pq},M_p,L_p,K$ as defined in this section.  The 
proof is given in section \ref{sec:proofs}.

\begin{lemma}
\label{lem:hp}
Define $\mathcal{S}=\mathrm{span}\{\chi_1,\ldots,\chi_n\}$, and
suppose that $\mu_j\chi_0 \in \mathcal{S}$ for all $j$.  Suppose moreover 
that $\mathcal{S}\subset\mathrm{Dom}(\lambda_j)$ and that
$\lambda_j\mathcal{S}\subset\mathcal{S}$ for all $j$.  Then the limit 
coefficients $N_{pq},M_p,L_p,K$ satisfy condition \ref{cond:limit}.
\end{lemma}

\begin{remark}
It is not difficult to construct examples of discrete interaction models 
where the limit coefficients $N_{pq},M_p,L_p,K$ do not satisfy the 
Hudson-Parthasarathy conditions in remark \ref{rem:hp}.  In this case, our 
proofs may be modified to show that $R_t^k$ nonetheless converges 
\textit{weakly} to the solution $U_t$ of equation (\ref{eq:hplim}) with 
coefficients $N_{pq},M_p,L_p,K$ as defined in this section.  However, 
in this case the limit evolution $U_t$ will not be unitary, so that the 
physical relevance of such a result is rather limited.
\end{remark}

\section{Examples}
\label{sec:ex}

We illustrate our results using four examples.  The first two examples
reproduce results of Attal and Pautrat \cite{AtP06} and Holevo
\cite{Hol92}.  The third example, that of a linear quantum system,
possesses both unbounded noise operators and unbounded initial
coefficients.  Finally, the fourth example shows that we may
simultaneously approximate the noise and the initial coefficients, as is
often useful in numerical simulations.

Until further notice $\mathcal{H}$ is a fixed Hilbert space.  Note that
for sake of example, all our models live in the Fock space with unit
multiplicity $n=1$.  The extension to multiple channels is entirely
straightforward and leads only to complication of the notation.

\subsection{Spin chain models}
\label{sec:spin}

Define $\mathcal{K} = \mathbb{C}^2$ and denote the canonical basis of 
$\mathbb{C}^2$ as $(\chi_0,\chi_1)$.  The noise Hilbert space is thus that 
of a single spin.  We embed the spin into the Fock space by defining
the embedding map $\pi^k:\mathcal{K}\to\mathcal{F}_{2^{-k}]}$ through
$$
	\pi^k\chi_0 = e(0) = 
	1\oplus \bigoplus_{p=1}^\infty 0,\qquad 
	\pi^k\chi_1 = 0\oplus 2^{k/2}\,I_{[0,2^{-k}]}\oplus
	\bigoplus_{p=2}^\infty 0
$$
(here $I_{[0,2^{-k}]}$ is the indicator function on the interval
$[0,2^{-k}]$).  We now define bounded self-adjoint operators 
$\lambda$, $\mu_1$ and $\mu_2$ on $\mathcal{K}$ as matrices
with respect to the canonical basis:
$$
  \lambda = \begin{pmatrix}
             1 & 0 \\
             0 & 0
            \end{pmatrix},\qquad
  \mu_1 =\sigma_x= \begin{pmatrix}
             0 & 1 \\
             1 & 0
            \end{pmatrix},\qquad
  \mu_2 = \sigma_y =\begin{pmatrix}
             0 & -i \\
             i & 0
            \end{pmatrix}.
$$
Let $F$, $G_1$, $G_2$ and $H$ be arbitrary bounded 
self-adjoint operators on $\mathcal{H}$, and let $H_K$ be a self-adjoint 
operator on $\mathcal{K}$.  Clearly the operator $H^k$ 
defined on $\mathcal{H}\otimes\mathcal{K}$ by
$$
	H^k = 2^k\,F\otimes\lambda +
	2^{k/2}\,(G_1\otimes\mu_1 + G_2\otimes\mu_2) + H\otimes I
	+ I\otimes H_K
$$
is self-adjoint for every $k$, and the domain assumptions in condition
\ref{cond:expl} are trivially satisfied.  Furthermore, we have 
$\lambda\chi_0 = 0$ and $\langle \chi_0, \mu_i\chi_0\rangle = 0$.

For every $\alpha \in \mathbb{C}$, we have defined $\chi^k(\alpha) = 
\chi_0 + 2^{-k/2}\alpha\chi_1$.  Note that
\begin{multline*}
	\|e(\alpha I_{[0,1]}) - (\pi^k\chi^k(\alpha))^{\otimes 2^k}\|^2 =
	\|e(\alpha I_{[0,2^{-k}]})\|^{2^{k+1}}+
	\|\pi^k\chi^k(\alpha)\|^{2^{k+1}} \\ -
	\langle e(\alpha I_{[0,2^{-k}]}),\pi^k\chi^k(\alpha)\rangle^{2^{k}} -
	\langle \pi^k\chi^k(\alpha),e(\alpha I_{[0,2^{-k}]})\rangle^{2^{k}} \\
	= e^{|\alpha|^2} - (1 + |\alpha|^2 2^{-k})^{2^k}
	\xrightarrow{k\to\infty}0.
\end{multline*}
We conclude that $(\pi^k\chi^k(\alpha))^{\otimes 2^k} 
\xrightarrow{k\to\infty}e(\alpha I_{[0,1]})$.  We have now verified all 
parts of condition 
\ref{cond:expl}.  Moreover, the coefficients $N_{11}$, $M_1$, $L_1$ and 
$K$ are easily computed:
\begin{equation*}
\begin{split}
  & N_{11} = e^{iF},\qquad M_1 = \frac{e^{iF}-I}{F}\,(G_1-iG_2),\qquad
  L_1 = (G_1+ iG_2)\,\frac{e^{iF}-I}{F}, \\
  &K = iH + i\langle\chi_0,H_K\chi_0\rangle\,I
  + (G_1 +iG_2)\,\frac{e^{iF} -iF -I}{F^2}\,(G_1-iG_2).
\end{split}
\end{equation*}
Note that these coefficients satisfy condition \ref{cond:limit} by virtue
of remark \ref{rem:hp}. It follows from theorem \ref{thm:expl} that the
repeated interaction $R_t^{k*}$ obtained from the Hamiltonian 
$H^k$ converges strongly to the solution $U_t^*$ of equation 
(\ref{eq:hplim}) with these coefficients, uniformly on compact
time intervals.  This result corresponds to \cite[theorem 19]{AtP06}.  

\subsection{Time-ordered exponentials}

Let $F,G,H$ be bounded operators on $\mathcal{H}$ with
$F,H$ self-adjoint, and define the essentially self-adjoint operators 
$$
	\Delta M^k_\ell := 
	F\,(\Lambda_{t_\ell}-\Lambda_{t_{\ell-1}})+
	G\,(A^\dag_{t_\ell}-A^\dag_{t_{\ell-1}})+
	G^*\,(A_{t_\ell}-A_{t_{\ell-1}})+
	H\,2^{-k}
$$
on $\mathcal{H}\oten\mathcal{E}$, where $t_\ell=\ell 2^{-k}$.  Holevo 
\cite{Hol92} defines time-ordered stochastic exponentials as the following 
strong limits on $\mathcal{H}\otimes\mathcal{F}$:
$$
	\overleftarrow{\mbox{exp}}\left[-i\int_0^t(F\,d\Lambda_s+
	G\,dA^\dag_s+G^*\,dA_s+H\,ds)\right] :=
	\mathop{\mbox{s-lim}}_{k\to\infty}\,
	e^{-i\Delta M_{\lfloor t2^k\rfloor}^k}\cdots 
	e^{-i\Delta M_1^k}.
$$
Evidently this definition can be interpreted as the limit of a sequence of 
discrete interaction models, and the limit process does indeed solve an
equation of the form (\ref{eq:hplim}).  In this example we develop this 
idea by applying theorem \ref{thm:expl}.

Define $\mathcal{K} = \mathcal{F}_{1]} = \Gamma_s(L^2([0,1]))$, and choose
the orthonormal vectors $\chi_0$ and $\chi_1$ as
\begin{equation*}
	  \chi_0 = e(0) = 1 \oplus \bigoplus_{p=1}^\infty 0,
	  \qquad \chi_1 = 0 \oplus
	  I_{[0,1]}\oplus\bigoplus_{p=2}^\infty 0.
\end{equation*}
We define embedding maps $\pi^k: \mathcal{K} \to \mathcal{F}_{2^{-k}]}$
by specifying their action on the (total) family of exponential vectors
$e(f)\in\mathcal{F}_{1]}$, $f\in L^2([0,1])$ as follows:
\begin{equation*}
	\pi^k e(f) = e(f^k),\qquad 
	f^k(x) = 2^{k/2} f(2^{k}x)
\end{equation*}
(note that $f^k\in L^2([0,2^{-k}])$).  It is easily verified that $\pi^k$ 
is a unitary map for every $k$.
We define self-adjoint operators $\mu_1,\mu_2$ and $\lambda$ on 
$\mathcal{K}$ as
\begin{equation*}
	  \mu_1 = A_1+A^\dag_1, \qquad
	  \mu_2 = i(A_1 -A^\dag_1), \qquad \lambda = \Lambda_1,
\end{equation*}
where $A_1, A^\dag_1$ and $\Lambda_1$
are the standard Hudson-Parthasarathy noises
evaluated at time $t=1$.  Define $G_1 = (G+G^*)/2$ and $G_2=i(G-G^*)/2$,
and set
\begin{equation*}
	  H^k := 2^k\,\pi^{k*}\Delta M^k_1\pi^k
	  = 2^k \,F \otimes\lambda  +
	  2^{k/2}\,(G_1\otimes\mu_1 + G_2\otimes\mu_2) + H\otimes I.
\end{equation*}
It is well known that $H^k$ and $\check F = F\otimes\lambda$ are 
essentially self-adjoint, and it is easily verified that 
$\langle\chi_0,\mu_j\chi_0\rangle=0$ and $\lambda\chi_0=0$.
Moreover, as $\pi^k\chi_0$ and $\pi^k\chi_1$ coincide with their 
counterparts in the previous example, we have verified that
condition \ref{cond:expl} holds.  The coefficients $N_{11}, M_1, L_1$ and 
$K$ are now easily computed:
\begin{equation*}
\begin{split}
 	& N_{11} = e^{iF},\qquad  M_1 = \frac{e^{iF}-I}{F}\,G,
  	\qquad L_1 = G^*\, \frac{e^{iF}-I}{F}, \\
	&  K = i H + G^*\, \frac{e^{iF}-I -  iF}{F^2}\,G.
\end{split}
\end{equation*}
These coefficients satisfy condition \ref{cond:limit} by virtue
of remark \ref{rem:hp}. It follows from theorem \ref{thm:expl} that the
time ordered exponential defined above coincides with the adjoint solution 
$U_t^*$ of equation (\ref{eq:hplim}) with these coefficients.
This agrees with \cite[corollary 1]{Hol92}.

\subsection{Linear quantum systems}

Let $\mathcal{H}=\ell^2(\mathbb{Z}_+)$ and denote by 
$(\eta_k,k\in\mathbb{Z}_+)$ the canonical basis in $\ell^2(\mathbb{Z}_+)$.
We also choose the domain 
$\mathcal{D}=\mathrm{span}\{\eta_k:k\in\mathbb{Z}_+\}\subset\mathcal{H}$
of finite particle vectors.  On $\mathcal{D}$ we define the operators
$$
	a\,\eta_k = \sqrt{k}\,\eta_{k-1},\qquad
	a^*\,\eta_k = \sqrt{k+1}\,\eta_{k+1},\qquad
	q = a+a^*,\qquad
	p = i(a-a^*).
$$
Note that $a$ is the annihilation operator and $a^*$ is the creation 
operator, while $q$ and $p$ are the position and momentum operators, 
respectively.

A linear quantum system is a quantum stochastic differential equation of 
the form (\ref{eq:hplim}) whose coefficients take the following form on
$\mathcal{D}$:
\begin{equation*}
\begin{split}
	&N_{11}=I,\quad
	M_1 = m\,p+m'\,q,\quad
	L_1 = -m^*\,p-{m'}^*\,q,\quad
	K = iH + \frac{1}{2}L_1M_1,\\
	& H = k_1p^2 + k_2(pq+qp)+k_3q^2 + k_4p+k_5q + k_6I,
\end{split}
\end{equation*}
where $m,m'\in\mathbb{C}$ and $k_1,\ldots,k_6\in\mathbb{R}$.  Physically,
a linear quantum system is a model whose Hamiltonian is quadratic in 
position and momentum and whose noise coefficients are linear in position 
and momentum.  At least formally, one may easily verify that the 
Hudson-Parthasarathy conditions of remark \ref{rem:hp} are satisfied, but 
the coefficients are unbounded in this case.  That condition 
\ref{cond:limit} is satisfied in this setting is proved in \cite{Fag90}.

Linear quantum systems possess various special properties: for example,
the adjoint solution $U_t^*$ of equation (\ref{eq:hplim}) leaves the 
family of Gaussian states in $\mathcal{H}\otimes\mathcal{F}$ invariant,
the Heisenberg evolution of the observables $(q,p)$ has an explicit 
solution, and the quantum filtering problem for (\ref{eq:hplim}) has a
finite-dimensional realization (the Kalman filter).  Because of these
and other properties, the linear quantum systems play a special role in 
quantum engineering as they admit particularly tractable methods for 
control synthesis and signal analysis \cite{JNP07,JNP07b,GSDM03}.  In 
these applications it could be of significant interest to work with 
discrete time approximations (e.g., for the purpose of digital signal 
processing), but it is important to seek approximations which preserve the 
linear systems properties of these models.  This is easily done, but we 
necessarily obtain discrete approximations where both the initial system 
coefficients and the discrete noises are unbounded.  In this example we 
will prove the convergence of such unbounded discrete approximations by 
appealing directly to our main theorem \ref{thm:main}.

We begin, however, by setting up our discrete models as in section 
\ref{sec:explicit}.  Let $\mathcal{K} = \ell^2(\mathbb{Z}_+)$, and choose 
an embedding $\pi^k:\mathcal{K}\to\mathcal{F}_{2^{-k}]}$ by setting
$$
        \pi^k\eta_0 = 
        1\oplus \bigoplus_{p=1}^\infty 0,\qquad
	\pi^k\eta_\ell = 
	\bigoplus_{p=0}^{\ell-1}0\oplus
	(2^{k/2}\,I_{[0,2^{-k}]})^{\otimes \ell}\oplus
	\bigoplus_{p=\ell+1}^{\infty}0
	\quad(\ell\in\mathbb{N}).
$$
If we define $\eta^k(\alpha):=\eta_0+2^{-k/2}\alpha\,\eta_1$ for 
$\alpha\in\mathbb{C}$, then we find precisely as in the previous examples
that $(\pi^k\eta^k(\alpha))^{\otimes 2^k}\to e(\alpha I_{[0,1]})$ as
$k\to\infty$.

Let us now define on $\mathcal{D}\oten\mathcal{D}$ the operators
$$
	H^k = H\otimes I
	-i\,(M_1\otimes a^* + L_1\otimes a)\,2^{k/2}.
$$
One may verify that $H^k$ is symmetric and that 
$\mathcal{D}\oten\mathcal{D}$ is a domain of analytic vectors for $H^k$, 
so that in particular $H^k$ is essentially self-adjoint for every $k$ 
\cite[section X.6]{ReS75}.  We will subsequently identify these operators 
with their closures.  Note that the Hamiltonian $H^k$ is quadratic in the 
family of position and momentum operators of the initial system and of the 
discrete noise; therefore the discrete interaction model is itself a 
(discrete time) linear system, and it therefore possesses all the 
associated desirable properties.

We now define the discrete interaction unitary $R^k$ on 
$\mathcal{H}\otimes\mathcal{F}$ from the Hamiltonian $H^k$ as in section 
\ref{sec:explicit}.  We will use theorem \ref{thm:main} to prove that
$$
	\lim_{k\to\infty}\sup_{0\le t\le T}
	\|R_t^{k*}\psi-U_t^*\psi\|=0\quad
	\mbox{for all }\psi\in\mathcal{H}\otimes\mathcal{F},~
	T<\infty,
$$
where $U_t$ is the solution of equation (\ref{eq:hplim}) with the
coefficients $N_{11},M_1,L_1,K$ defined above.
Note that theorem \ref{thm:expl} does not apply as the initial 
coefficients are unbounded, but we may essentially repeat the proof of 
that theorem with minor modifications to obtain the present result.
To this end, we begin by noting that $\mathcal{D}$ is a core for
$\mathscr{L}^{\alpha\beta}$ ($\alpha,\beta\in\mathbb{C}$) by the analytic 
vector theorem (see \cite[remark 4]{BHS08}).  Let us fix 
$\alpha,\beta\in\mathbb{C}$, and define $\psi^k=\pi^k\eta^k(\alpha)$ and
$\varphi^k=\pi^k\eta^k(\beta)$.  By theorem \ref{thm:main}, it suffices to 
prove that
$$
	\|2^k(R^{k;\psi^k\varphi^k}-I)u-\mathscr{L}^{\alpha\beta}u\|
	\xrightarrow{k\to\infty}0
$$
for every $u\in\mathcal{D}$.  We now proceed as follows.
Fix $u\in\mathcal{D}$, and note that using the trivial identities $e^{ix} 
= 1 + f(x)\,x = 1 + ix + g(x)\,x^2$ we can write
\begin{multline*}
	e^{iH^k2^{-k}}\,u\otimes\eta_\ell =
	(I+2^{-k}\,f(H^k2^{-k})\,H^k)\,u\otimes\eta_\ell \\ = 
	(I+iH^k2^{-k}+2^{-2k}\,g(H^k2^{-k})\,(H^k)^2)\,u\otimes\eta_\ell.
\end{multline*}
Here we have used the spectral theorem and the fact that 
$u\otimes\eta_\ell\in\mathrm{Dom}((H^k)^p)$ for every $\ell,p$
(see the proof of theorem \ref{thm:expl} for a more precise argument).
Therefore
\begin{equation*}
\begin{split}
	&\langle v\otimes\eta^k(\alpha),2^k(e^{iH^k2^{-k}}-I)\,
	u\otimes\eta^k(\beta)\rangle = 
	i\langle v,Hu\rangle  \\
	&\qquad\qquad
	- \langle v\otimes\eta_0,g(H^k2^{-k})\,
	(M_1\otimes a^*+L_1\otimes a)\,M_1u\otimes\eta_1\rangle \\
	&\qquad\qquad
	-i\,\alpha^*\,
	\langle v\otimes\eta_1,f(H^k2^{-k})\,M_1u\otimes\eta_1\rangle
	\\
	&\qquad\qquad
	-i\,
	\langle v\otimes\eta_0,f(H^k2^{-k})\,(M_1\otimes a^*+L_1\otimes a)
	\,u\otimes\eta_1\rangle
	\,\beta \\
	&\qquad\qquad
	+ \alpha^*
	\langle v\otimes\eta_1,(e^{iH^k2^{-k}}-I)\,
        u\otimes\eta_1\rangle\,\beta
	+ O(\|v\|\,2^{-k/2}).	
\end{split}
\end{equation*}
A straightforward computation shows that
$$
	\sup_{\stackrel{v\in\mathcal{H}}{\|v\|\le 1}}
	|\langle v\otimes\eta^k(\alpha),2^k(e^{iH^k2^{-k}}-I)\,
        u\otimes\eta^k(\beta)\rangle -
	\langle v, (\alpha^*\,M_1+L_1\,\beta+K)
	\,u\rangle|\xrightarrow{k\to\infty}0,
$$
where we have used that $f(H^k2^{-k})\to iI$,
$g(H^k2^{-k})\to -(1/2)I$ and $e^{iH^k2^{-k}}\to I$ strongly 
as $k\to\infty$ by \cite[theorems VIII.20 and VIII.25]{ReS80}.
As in the proof of theorem \ref{thm:expl} below, it follows
readily that 
$\|2^k(R^{k;\psi^k\varphi^k}-I)u-\mathscr{L}^{\alpha\beta}u\|\to 0$.

\subsection{Finite dimensional approximations}

In the previous examples we have approximated the quantum stochastic
differential equation (\ref{eq:hplim}) by constructing discrete
interaction models whose interaction unitary $R^k$ lives on the Hilbert
space $\mathcal{H}\otimes\mathcal{K}^k$.  Even though the Fock space
$\mathcal{F}$ is infinite dimensional, we have seen that we may choose
finite-dimensional discrete noise spaces as simple as
$\mathcal{K}^k\cong\mathbb{C}^2$.  In numerical applications, however, we
typically wish to go one step further and approximate also the initial
space $\mathcal{H}$ (which is often infinite dimensional) by a finite
dimensional space $\mathcal{H}^k$.  The discrete interaction models then
live entirely on the finite dimensional Hilbert spaces
$\mathcal{H}^k\otimes\mathcal{K}^k$, as is desirable for numerical 
implementation.  We would like to establish that these discrete models 
converge to the solution of the limit equation (\ref{eq:hplim}) when we 
simultaneously let the time step go to zero and let the initial space 
dimension go to infinity.  We will now show in a toy example that this 
problem fits into the setting of theorem \ref{thm:main}.

We will discretize the noise essentially as in the example of section 
\ref{sec:spin}, but let us directly embed the discrete models (rather than 
work with the embedding maps $\pi^k$) to simplify the notation.
For every $k\in\mathbb{N}$, define the vectors $\chi^k_0,\chi^k_1\in
\mathcal{F}_{2^{-k}]}$ as
$$
	\chi_0^k =
	1\oplus \bigoplus_{p=1}^\infty 0,\qquad 
	\chi_1^k = 0\oplus 2^{k/2}\,I_{[0,2^{-k}]}\oplus
	\bigoplus_{p=2}^\infty 0.
$$
Moreover, let $\mathcal{H}$ be an infinite dimensional separable Hilbert 
space, and fix an orthonormal basis 
$\{\eta_k:k\in\mathbb{N}\}\subset\mathcal{H}$.  For the $k$th discrete 
interaction model we will choose the initial Hilbert space
$\mathcal{H}^k:=\mathrm{span}\{\eta_1,\ldots,\eta_k\}\subset\mathcal{H}$ 
and the noise Hilbert space $\mathcal{K}^k:=\mathrm{span}\{\chi_0^k,
\chi_1^k\}\subset\mathcal{F}_{2^{-k}]}$.  Thus 
$\mathcal{H}^k\otimes\mathcal{K}^k$ is indeed finite dimensional.

Let $H$ and $M$ be bounded operators on $\mathcal{H}$ where $H$ is 
self-adjoint.  In this toy example, we will be interested in approximating 
the solution of the limit equation
$$
	dU_t = U_t\,\{M\,dA_t^\dag - M^*\,dA_t + iH\,dt 
	- \tfrac{1}{2}M^*M\,dt\},
$$
which satisfies condition \ref{cond:limit} by virtue
of remark \ref{rem:hp}.  We claim that this can be done by choosing the 
discrete interaction unitaries
$$
	R^k = \exp\{
		2^{-k/2}\,(P_kMP_k\otimes b_k^* - P_kM^*P_k\otimes b_k)
		+ i\,2^{-k}\,P_kHP_k\otimes I
	\},
$$
where $P_k$ denotes the orhogonal projection onto $\mathcal{H}^k$ and
the noise operator $b_k$ is defined by setting $b_k\chi_0^k=0$, 
$b_k\chi_1^k=\chi_0^k$, and $b_k\psi=0$ for $\psi\perp\mathcal{K}^k$.

We simply verify the conditions of theorem \ref{thm:main}.  As all the 
coefficients are bounded, we may choose 
$\mathcal{D}^{\alpha\beta}=\mathcal{D}=\mathcal{H}$.  Fix 
$\alpha,\beta\in\mathbb{C}$ and $u\in\mathcal{H}$, and let us define
$\psi^k=\chi_0^k+2^{-k/2}\alpha\chi_1^k$ and 
$\varphi^k=\chi_0^k+2^{-k/2}\beta\chi_1^k$.  As in the previous examples 
we have 
$$
	(\psi^{k})^{\otimes 2^k}\xrightarrow{k\to\infty}
	e(\alpha I_{[0,1]}),\qquad
	(\varphi^{k})^{\otimes 2^k}\xrightarrow{k\to\infty}
	e(\beta I_{[0,1]}).
$$
As $P_ku\to u$ as $k\to\infty$, it suffices to verify that
$$
	2^k(R^{k;\psi^k\varphi^k}-I)P_ku
	\xrightarrow{k\to\infty}\mathscr{L}^{\alpha\beta}u.
$$
What remains is again essentially the same computation as in the previous 
example and as in the proof of theorem \ref{thm:expl}.  We leave the 
details to the reader.

\section{Proofs}
\label{sec:proofs}

\subsection{Proof of theorem \ref{thm:main}}

The proof of theorem \ref{thm:main} is based on a version of the 
Trotter-Kato theorem for contraction semigroups due to Kurtz.  We cite 
here from \cite[theorem 1.6.5]{EK86} a special case of this result in the 
form which will be convenient in the following.

\begin{theorem}[Kurtz]
\label{thm:kurtz}
Let $\mathcal{H}$ be a fixed Hilbert space.  For $k\in\mathbb{N}$, let 
$T^k$ be a linear contraction on $\mathcal{H}$ and let $T_t$ be a strongly 
continuous contraction semigroup on $\mathcal{H}$ with generator
$\mathscr{L}$.  Let $\mathcal{D}$ be a core for $\mathscr{L}$.  Then the 
following are equivalent.
\begin{enumerate}
\item For every $u\in\mathcal{D}$, there exist $u^k\in\mathcal{H}$ such
that
$$
	u^k\xrightarrow{k\to\infty}u,\qquad
	2^k(T^k-I)u^k\xrightarrow{k\to\infty}\mathscr{L}u.
$$
\item For every $\psi\in\mathcal{H}$ and $t<\infty$
$$
	\lim_{k\to\infty}
	\|(T^k)^{\lfloor t2^k\rfloor}\psi-T_t\psi\|=0.
$$
\item For every $\psi\in\mathcal{H}$ and $t<\infty$
$$
	\lim_{k\to\infty}\sup_{s\le t}
	\|(T^k)^{\lfloor s2^k\rfloor}\psi-T_s\psi\|=0.
$$
\end{enumerate}
\end{theorem}

Armed with this result, we may now proceed to prove theorem 
\ref{thm:main}.  We will first prove the forward direction, and we
subsequently consider the converse implication.

\begin{proof}[Theorem \ref{thm:main}, 1$\Rightarrow$2]
Let us restrict our attention to the interval $[0,N]$ with
$N\in\mathbb{N}$.  It suffices to prove that convergence holds uniformly
on $[0,N]$ for any $N$.  We may therefore restrict the Hilbert space to
$\mathcal{H}\otimes\mathcal{F}_{N]}$, which we do from now on.

First, let $\alpha,\beta\in\mathbb{C}^n$ and $\psi^k,\varphi^k\in
\mathcal{F}_{2^{-k}]}$ as in the statement of the theorem.  Then
\begin{multline*}
	\langle u\otimes(\psi^k)^{\otimes N2^k},
	R_t^k\,v\otimes(\varphi^k)^{\otimes N2^k}\rangle =
	\|\psi^k\|^{N2^k}\|\varphi^k\|^{N2^k}
	\langle u,(R^{k;\psi^k\varphi^k})^{\lfloor t2^k\rfloor}v\rangle \\
	\xrightarrow{k\to\infty}
	e^{(|\alpha|^2+|\beta|^2)N/2}
	\langle u,T_t^{\alpha\beta}v\rangle
	=
	\langle u\otimes e(\alpha I_{[0,N]}),
	U_t\,v\otimes e(\beta I_{[0,N]})\rangle
\end{multline*}
for any $u,v\in\mathcal{H}$ and $t\in[0,N]$ by lemma \ref{lem:discgen} and
theorem \ref{thm:kurtz}.  Similarly, we can establish the following.
Let $0=t_0<t_1<\cdots<t_m<t_{m+1}=N$ ($m\in\mathbb{N}$) be a dyadic
rational partition of $[0,N]$, i.e., $t_j=\ell_j2^{-k'}$ for
some $k'\in\mathbb{N}$ with $\ell_j\in\mathbb{N}$ for all $j=1,\ldots,m$.  
Let $\alpha_0,\ldots,\alpha_m,\beta_0,\ldots,\beta_m\in\mathbb{C}^n$, and
choose $\psi_j^k,\varphi_j^k\in\mathcal{F}_{2^{-k}]}$ such that
$$
	(\psi^{k}_j)^{\otimes 2^k}\xrightarrow{k\to\infty}
	e(\alpha_j I_{[0,1]}),\qquad
	(\varphi^{k}_j)^{\otimes 2^k}\xrightarrow{k\to\infty}
	e(\beta_j I_{[0,1]}).
$$
Let $f,g\in L^2([0,N];\mathbb{C}^n)$ be simple functions with
$f(s)=\alpha_{j}$ and $g(s)=\beta_j$ for $s\in [t_j,t_{j+1}[\mbox{}$, and
define for all $k\ge k'$ the vectors
\begin{equation*}
\begin{split}
	&\psi_k=(\psi^k_0)^{\otimes \ell_12^{k-k'}}
	\otimes (\psi^k_1)^{\otimes (\ell_2-\ell_1)2^{k-k'}}
	\otimes\cdots
	\otimes (\psi^k_m)^{\otimes (N2^{k'}-\ell_m)2^{k-k'}},\\
	&\varphi_k=(\varphi^k_0)^{\otimes \ell_12^{k-k'}}
	\otimes (\varphi^k_1)^{\otimes (\ell_2-\ell_1)2^{k-k'}}
	\otimes\cdots
	\otimes (\varphi^k_m)^{\otimes (N2^{k'}-\ell_m)2^{k-k'}}.
\end{split}
\end{equation*}
Then $\psi_k\to e(f)$ and $\varphi_k\to e(g)$ as $k\to\infty$, and
it is not difficult to verify (using the cocycle property of $U_t$)
that for $t\in [t_j,t_{j+1}[\mbox{}$
$$
        \langle u\otimes e(f),U_t\,v\otimes e(g)\rangle =
        \|e(f)\|\,\|e(g)\|\,
        \langle u,T_{t_1}^{\alpha_{0}\beta_{0}}
                T_{t_2-t_1}^{\alpha_{1}\beta_{1}}\cdots
                T_{t-t_j}^{\alpha_{j}\beta_{j}}v\rangle.
$$
We therefore obtain for $u,v\in\mathcal{H}$ and 
$t\in [t_j,t_{j+1}[\mbox{}$
\begin{multline*}
	\langle u\otimes\psi_k,R_t^k\,v\otimes\varphi_k\rangle =
	\|\psi_k\|\,\|\varphi_k\|\,
	\langle u,(R^{k;\psi^k_0\varphi^k_0})^{t_12^k}
	\cdots
	(R^{k;\psi^k_j\varphi^k_j})^{\lfloor
	(t-t_j)2^k
	\rfloor}v\rangle \\
	\xrightarrow{k\to\infty}
        \|e(f)\|\,\|e(g)\|\,
        \langle u,T_{t_1}^{\alpha_{0}\beta_{0}}
                T_{t_2-t_1}^{\alpha_{1}\beta_{1}}\cdots
                T_{t-t_j}^{\alpha_{j}\beta_{j}}v\rangle =
        \langle u\otimes e(f),U_t\,v\otimes e(g)\rangle,
\end{multline*}
where we have again appealed to lemma \ref{lem:discgen} and theorem 
\ref{thm:kurtz}.  Now note that
\begin{multline*}
	|\langle u\otimes\psi_k,R_t^k\,v\otimes\varphi_k\rangle -
	\langle u\otimes e(f),R_t^k\,v\otimes e(g)\rangle| \\ \le
	|\langle u\otimes(\psi_k-e(f)),R_t^k\,v\otimes\varphi_k\rangle|
	+
	|\langle u\otimes e(f),R_t^k\,v\otimes(\varphi_k-e(g))\rangle
	\\
	\le \|u\|\,\|v\|\,(\|\psi_k-e(f)\|\,\|\varphi_k\|+
		\|e(f)\|\,\|\varphi_k-e(g)\|
	\xrightarrow{k\to\infty}0,
\end{multline*}
where we have used that $R_t^k$ is unitary.  Therefore
$$
	\langle u\otimes e(f),R_t^k\,v\otimes e(g)\rangle
	\xrightarrow{k\to\infty} 
	\langle u\otimes e(f),U_t\,v\otimes e(g)\rangle
$$
for any $u,v\in\mathcal{H}$ and simple functions $f,g\in 
L^2([0,N];\mathbb{C}^n)$ with dyadic rational jump points.
As the latter are dense in $L^2([0,N];\mathbb{C}^n)$, a
similar approximation argument shows that $R_t^k$ converges
weakly to $U_t$ as $k\to\infty$ for every $t\in[0,N]$.

It remains to strengthen weak convergence to strong convergence uniformly
on $[0,N]$.  First, note that as $R_t^k$ and $U_t$ are all unitary,
weak convergence of $R_t^k$ to $U_t$ (which is of course equivalent to 
weak convergence of $(R_t^k)^*$ to $U_t^*$) already implies that 
$(R_t^k)^*\to U_t^*$ strongly as $k\to\infty$ for every fixed time 
$t\in[0,N]$.  To prove that the convergence is in fact uniform, we will 
utilize another implication of theorem \ref{thm:kurtz}.

It is convenient to extend the Fock space to two-sided time,
i.e., we will consider the ampliations of all our operators to the
extended Fock space $\mathcal{\tilde F}=\Gamma_s(L^2(\mathbb{R};\mathbb{C}^n))
\cong\mathcal{F}_-\otimes\mathcal{F}$, where $\mathcal{F}_-\cong\mathcal{F}$
is the negative time portion of the two-sided Fock space.  We now define
the two-sided shift $\tilde\theta_t:L^2(\mathbb{R};\mathbb{C}^n)\to
L^2(\mathbb{R};\mathbb{C}^n)$ as $\tilde\theta_tf(s)=f(t+s)$, and by
$\tilde\Theta_t:\mathcal{\tilde F}\to\mathcal{\tilde F}$ its second
quantization.  Note that $\tilde\Theta_t$ is a strongly continuous
one-parameter unitary group, and that the cocycle property reads
$U_{t+s}=U_s\tilde\Theta_s^*U_t\tilde\Theta_s$, etc., in terms of the
two-sided shift.  Now define on the two-sided Fock space the operators
$$
        V_t=\tilde\Theta_tU_t^*,\qquad\qquad
	S^k=\tilde\Theta_{2^{-k}}(R^k)^*.
$$
Then it is immediate from the cocycle property that $V_t$
defines a strongly continuous unitary (hence contraction) semigroup on
$\mathcal{H}\otimes\mathcal{\tilde F}$ and that the unitary (hence 
contraction) $S^k$ is 
such that $(S^k)^{\lfloor t2^k\rfloor}=\tilde\Theta_{2^{-k}\lfloor 
t2^k\rfloor}(R_t^k)^*$.  Moreover, for any dyadic rational $t$
$$
        \|(R_t^{k})^*\psi - U_t^*\psi\| =
        \frac{\|(S^k)^{\lfloor t2^k\rfloor}\,\psi_-\otimes\psi 
	- V_t\,\psi_-\otimes\psi\|}
        {\|\psi_-\|}\qquad
        \forall\,\psi\in\mathcal{H}\otimes\mathcal{F},
        ~\psi_-\in\mathcal{F}_-
$$
for $k$ sufficiently large, as $\tilde\Theta_t$ is an isometry (here
$\psi_-\otimes\psi\in\mathcal{F}_-\otimes\mathcal{H}\otimes\mathcal{F}
\cong\mathcal{H}\otimes\mathcal{\tilde F}$).  As vectors of the
form $\psi_-\otimes\psi$ are total in $\mathcal{H}\otimes\mathcal{\tilde 
F}$ and as we have already established that
$\|(R_t^{k})^*\psi - U_t^*\psi\|\to 0$ as $k\to\infty$, we find that for 
every fixed dyadic rational $t$
$$
	\|(S^k)^{\lfloor t2^k\rfloor}\,\psi
	- V_t\,\psi\|\xrightarrow{k\to\infty}0
	\quad\mbox{for all }\psi\in\mathcal{H}\otimes\mathcal{\tilde F}.
$$
As $V_t$ is strongly continuous and the dyadic rationals are dense, this 
evidently holds for every fixed $t$.  It remains to apply the implication 
2$\Rightarrow$3 of theorem \ref{thm:kurtz}.
\end{proof}

\begin{proof}[Theorem \ref{thm:main}, 2$\Rightarrow$1]
Fix $\alpha,\beta\in\mathbb{C}^n$, $v\in\mathcal{H}$, and
$t\in[0,N]$, and let us choose the sequences $\psi^k=e(\alpha 
I_{[0,2^{-k}]})$ and $\varphi^k=e(\beta I_{[0,2^{-k}]})$.  Then we can 
estimate
\begin{equation*}
\begin{split}
        &\|(R^{k;\psi^k\varphi^k})^{\lfloor t2^k\rfloor}v
		-T_t^{\alpha\beta}v\|^2
	=
        \sup_{\stackrel{u\in\mathcal{H}}{\|u\|\le 1}}
        |\langle u,(R^{k;\psi^k\varphi^k})^{\lfloor 
		t2^k\rfloor}v-T_t^{\alpha\beta}v\rangle|^2
        \\ 
	&\qquad =
	\sup_{\stackrel{u\in\mathcal{H}}{\|u\|\le 1}}
        \frac{|\langle u\otimes e(\alpha I_{[0,N]}),
        (R_t^{k}-U_t)\,v\otimes e(\beta I_{[0,N]})\rangle|^2}
        {e^{(|\alpha|^2+|\beta|^2)N}}  
	\\
	&\qquad \le
        \frac{\|(R_t^{k}-U_t)\,v\otimes e(\beta I_{[0,N]})\|^2}
	{e^{|\beta|^2N}} \\
	&\qquad 
        =
        \frac{2\,\|v\otimes e(\beta I_{[0,N]})\|^2
        - 2\,\mathrm{Re}(\langle R_t^{k}\,v\otimes e(\beta I_{[0,N]})
	,U_t\,v\otimes e(\beta I_{[0,N]})\rangle)}
	{e^{|\beta|^2N}}.
\end{split}
\end{equation*}
But by assumption $\|(R_t^{k*}- U_t^*)\,\psi\|\to 0$ as $k\to\infty$
for all $\psi\in\mathcal{H}\otimes\mathcal{F}$, so that in particular
$\langle R_t^{k}\,v\otimes e(\beta I_{[0,N]}),U_t\,v\otimes e(\beta 
I_{[0,N]})\rangle\to \|v\otimes e(\beta I_{[0,N]})\|^2$.  We 
thus obtain 
$$
        \|(R^{k;\psi^k\varphi^k})^{\lfloor t2^k\rfloor}v
	-T_t^{\alpha\beta}v\| 
        \xrightarrow{k\to\infty}0\qquad
        \mbox{for all }v\in\mathcal{H},~t\ge 0.
$$
The result now follows by appealing to the implication 
2$\Rightarrow$1 of theorem \ref{thm:kurtz}. 
\end{proof}

\subsection{Proof of theorem \ref{thm:expl}}

The proof of theorem \ref{thm:expl} is chiefly a matter of straightforward 
computation.  We make use of one simple trick: the trivial identities
$$
	e^{ix} = 1 + x\,f(x) = 1 + ix + x^2\,g(x)
$$
and the conditions $\lambda_j\chi_0=0$, $\langle\chi_0,\mu_j\chi_0\rangle=0$ 
allow us to cancel those terms in the expression for 
$2^k(R^{k;\psi^k\varphi^k}-I)u$ which diverge as $k\to\infty$.

\begin{proof}[Theorem \ref{thm:expl}]
Throughout the proof we fix $\alpha,\beta\in\mathbb{C}^n$ and
$u\in\mathcal{H}$, and we define $\psi^k=\pi^k\chi^k(\alpha)$, 
$\varphi^k=\pi^k\chi^k(\beta)$.  By theorem \ref{thm:main}, it suffices
to show that
$$
	\|2^k(R^{k;\psi^k\varphi^k}-I)u-\mathscr{L}^{\alpha\beta}u\|
	\xrightarrow{k\to\infty}0,
$$
where $\mathscr{L}^{\alpha\beta}u$ is given in lemma \ref{lem:gens} in 
terms of the coefficients defined in section \ref{sec:explicit}.

For the time being, let us fix $k$.  As $H^k$ is self-adjoint, we may
assume by the spectral theorem that $\mathcal{H}\otimes\mathcal{K}\cong
L^2(\Omega^k)$ for some measure space $(\Omega^k,\Sigma^k,P^k)$ and that
$H^k$ acts on $L^2(\Omega^k)$ by pointwise multiplication
$(H^k\psi)(\omega)=h^k(\omega)\psi(\omega)$ for all $\psi\in
L^2(\Omega^k)$.  We represent the vectors $u\otimes\chi_j$ and
$v\otimes\chi_j$ in $L^2(\Omega^k)$ as $u_j(\omega)$ and $v_j(\omega)$, 
respectively.  As $u\otimes\chi_0$ and $v\otimes\chi_0$ are in 
$\mathrm{Dom}(H^k)$ (so that $h^k(\omega)u_0(\omega)$ and 
$h^k(\omega)v_0(\omega)$ are square integrable), the trivial identity 
$e^{ix}=1+ix+x^2\,g(x)$ gives that
\begin{multline*}
	\langle v\otimes\chi_0,2^k(e^{iH^k2^{-k}}-I)
	\,u\otimes\chi_0\rangle =
	i\int v_0(\omega)^*h^k(\omega)u_0(\omega)\,P^k(d\omega) \\
	+2^{-k}\int 
	(h^k(\omega)v_0(\omega))^*\,g(h^k(\omega)2^{-k})\,h^k(\omega)
	u_0(\omega)\,P^k(d\omega) \\
	= i\langle v\otimes\chi_0,H^k\,u\otimes\chi_0\rangle +
	2^{-k}\langle H^k\, v\otimes\chi_0,g(H^k2^{-k})\,
	H^k\,u\otimes\chi_0\rangle.
\end{multline*}
Using the trivial identity $e^{ix}=1+x\,f(x)$, we similarly
obtain for $p=1,\ldots,n$
\begin{equation*}
\begin{split}
	&
	\langle v\otimes\chi_p,2^k(e^{iH^k2^{-k}}-I)
	\,u\otimes\chi_0\rangle =
	\langle v\otimes\chi_p,f(H^k2^{-k})\,H^k\,u\otimes\chi_0\rangle,\\
	&
	\langle v\otimes\chi_0,2^k(e^{iH^k2^{-k}}-I)
	\,u\otimes\chi_p\rangle =
	\langle H^k\,v\otimes\chi_0,f(H^k2^{-k})\,u\otimes\chi_p\rangle.
\end{split}
\end{equation*}
Using that $\lambda_j\chi_0=0$ and $\langle\chi_0,\mu_j\chi_0\rangle=0$,
a simple computation gives the following.
\begin{equation*}
\begin{split}
	&\langle v\otimes\chi^k(\alpha),2^k(e^{iH^k2^{-k}}-I)\,
	u\otimes\chi^k(\beta)\rangle = \\
	&\qquad\qquad 
	i\sum_{j=1}^r\langle v\otimes\chi_0,H_ju\otimes\nu_j\chi_0\rangle \\
	&\qquad\qquad
	+
	\sum_{j,j'=1}^m\langle G_jv\otimes\mu_j\chi_0,g(H^k2^{-k})\,
	G_{j'}u\otimes\mu_{j'}\chi_0\rangle \\
	&\qquad\qquad
	+ \sum_{p=1}^n\alpha_p^*\sum_{j=1}^m
	\langle v\otimes\chi_p,f(H^k2^{-k})\,G_ju\otimes\mu_j\chi_0\rangle
	\\
	&\qquad\qquad
	+ \sum_{p=1}^n\sum_{j=1}^m
	\langle G_jv\otimes\mu_j\chi_0,f(H^k2^{-k})\,u\otimes\chi_p\rangle
	\,\beta_p \\
	&\qquad\qquad
	+ \sum_{p,q=1}^n\alpha^*_p\,
	\langle v\otimes\chi_p,(e^{iH^k2^{-k}}-I)\,
        u\otimes\chi_q\rangle\,\beta_p
	+ O(\|v\|\,2^{-k/2}).
\end{split}
\end{equation*}
To proceed, let us define 
\begin{equation*}
        \overline{\mathscr{L}}^{\alpha\beta}u =
        \left(\sum_{p,q=1}^n \alpha_p^* (N_{pq}-\delta_{pq})\beta_q +
        \sum_{p=1}^n \alpha^*_p M_p +
        \sum_{p=1}^n L_p \beta_p + K
        \right)u.
\end{equation*}
Using the definition of $N_{pq},M_p,L_p,K$ in section \ref{sec:explicit},
we obtain
\begin{equation*}
\begin{split}
	&\langle v,\overline{\mathscr{L}}^{\alpha\beta}u\rangle =
	i\sum_{j=1}^r\langle v\otimes\chi_0,H_ju\otimes\nu_j\chi_0\rangle 
	\\ &\qquad\qquad
	+
	\sum_{j,j'=1}^m\langle G_jv\otimes\mu_j\chi_0,g(\check F)\,
	G_{j'}u\otimes\mu_{j'}\chi_0\rangle \\
	&\qquad\qquad
	+ \sum_{p=1}^n\alpha_p^*\sum_{j=1}^m
	\langle v\otimes\chi_p,f(\check F)\,G_ju\otimes\mu_j\chi_0\rangle
	\\
	&\qquad\qquad
	+ \sum_{p=1}^n\sum_{j=1}^m
	\langle G_jv\otimes\mu_j\chi_0,f(\check F)\,u\otimes\chi_p\rangle
	\,\beta_p \\
	&\qquad\qquad
	+ \sum_{p,q=1}^n\alpha^*_p\,
	\langle v\otimes\chi_p,(e^{i\check F}-I)\,
        u\otimes\chi_q\rangle\,\beta_p.
\end{split}
\end{equation*}
Thus evidently we obtain
$$
	\sup_{\stackrel{v\in\mathcal{H}}{\|v\|\le 1}}
	|\langle v\otimes\chi^k(\alpha),2^k(e^{iH^k2^{-k}}-I)\,
        u\otimes\chi^k(\beta)\rangle-
	\langle v,\overline{\mathscr{L}}^{\alpha\beta}u\rangle|
	\xrightarrow{k\to\infty}0,
$$
provided that $f(H^k2^{-k})\to f(\check F)$, $g(H^k2^{-k})\to g(\check F)$
and $e^{iH^k2^{-k}}\to e^{i\check F}$ strongly as $k\to\infty$.  This 
is indeed the case due to \cite[theorems VIII.20 and VIII.25]{ReS80}.

To complete the proof, note that we can write
\begin{equation*}
\begin{split}
	&\|2^k(R^{k;\psi^k\varphi^k}-I)u-\mathscr{L}^{\alpha\beta}u\| 
	= \sup_{\stackrel{v\in\mathcal{H}}{\|v\|\le 1}}|
	\langle v,2^k(R^{k;\psi^k\varphi^k}-I)u\rangle -
	\langle v,\mathscr{L}^{\alpha\beta}u\rangle| \\
	&\qquad
	= \sup_{\stackrel{v\in\mathcal{H}}{\|v\|\le 1}}\left|
	\frac{2^k\,\langle v\otimes\chi^k(\alpha),e^{iH^k2^{-k}}\,
        u\otimes\chi^k(\beta)\rangle}{\|\chi^k(\alpha)\|\,
	\|\chi^k(\beta)\|} 
	- 2^k\,\langle v,u\rangle
	- \langle v,\mathscr{L}^{\alpha\beta}u\rangle
	\right|
	\\
	&\qquad
	\le
	\frac{
	\sup_{v\in\mathcal{H},\|v\|\le 1}
	|\langle v\otimes\chi^k(\alpha),2^k(e^{iH^k2^{-k}}-I)\,
        u\otimes\chi^k(\beta)\rangle-
	\langle v,\overline{\mathscr{L}}^{\alpha\beta}u\rangle|
	}{\|\chi^k(\alpha)\|\,\|\chi^k(\beta)\|}
	\\
	&\qquad\qquad
	+ \left\|
	2^k\left(
	\frac{\langle\chi^k(\alpha),\chi^k(\beta)\rangle}{
		\|\chi^k(\alpha)\|\,\|\chi^k(\beta)\|}-1\right)u
	+
	\frac{\overline{\mathscr{L}}^{\alpha\beta}u}{
		\|\chi^k(\alpha)\|\,\|\chi^k(\beta)\|}	
	- \mathscr{L}^{\alpha\beta}u
	\right\|,
\end{split}
\end{equation*}
which converges to zero as $k\to\infty$.  The claim has been established.
\end{proof}

\subsection{Proof of lemma \ref{lem:uniqueness}}

By our assumptions, we have
$$
	\langle \chi^k(\alpha),\tilde\chi^k(\beta)\rangle^{2^k} =
	\langle (\pi^k\chi^k(\alpha))^{\otimes 2^k},
		(\pi^k\tilde\chi^k(\beta))^{\otimes 
		2^k}\rangle 
	\xrightarrow{k\to\infty} 
	e^{\alpha^*\beta}
$$
for every $\alpha,\beta\in\mathbb{C}^n$.  For $\alpha,\beta=0$ we find 
that $\langle\chi_0,\tilde\chi_0\rangle^{2^k}\to 1$ as $k\to\infty$, so we 
must have $|\langle\chi_0,\tilde\chi_0\rangle|=1$.  But $\chi_0$ and 
$\tilde\chi_0$ are unit vectors, so $\tilde\chi_0=e^{i\phi}\chi_0$ 
for some $\phi\in\mathbb{R}$ such that $e^{i\phi 2^k}\to 1$ as 
$k\to\infty$.  We may now compute for arbitrary 
$\alpha,\beta\in\mathbb{C}^n$
\begin{multline*}
	\langle \chi^k(\alpha),\tilde\chi^k(\beta)\rangle^{2^k} =
	e^{i\phi 2^k}
	\left[1 + 2^{-k}e^{-i\phi}\sum_{i,j=1}^n\alpha_i^*\,\langle
	\chi_i,\tilde\chi_j\rangle\,\beta_j\right]^{2^k} \\
	\mbox{}
	\xrightarrow{k\to\infty}
	\exp\left[e^{-i\phi}\sum_{i,j=1}^n\alpha_i^*\,\langle
        \chi_i,\tilde\chi_j\rangle\,\beta_j\right].
\end{multline*}
Substituting $\alpha\mapsto t\alpha$ and differentiating, we find that
$$
	0 = \frac{d}{dt}\left.\left[
	e^{te^{-i\phi}\sum_{i,j=1}^n\alpha_i^*\,\langle
        \chi_i,\tilde\chi_j\rangle\,\beta_j}-
	e^{t\alpha^*\beta}
	\right]\right|_{t=0} =
	e^{-i\phi}\sum_{i,j=1}^n\alpha_i^*\,\langle
        \chi_i,\tilde\chi_j\rangle\,\beta_j -
	\alpha^*\beta
$$
for every $\alpha,\beta\in\mathbb{C}^n$.  We thus find that
$\langle\chi_j,\tilde\chi_j\rangle=e^{i\phi}$ for every $j$, which
implies that $\tilde\chi_j=e^{i\phi}\chi_j$ as $\chi_j$ and 
$\tilde\chi_j$ are unit vectors.  The proof is complete.
\qed

\subsection{Proof of lemma \ref{lem:hp}}

As the coefficients are bounded, it suffices to verify the 
Hudson-Parthasarathy conditions in remark \ref{rem:hp}.
We first show that
$$
  	\sum_{j=1}^n N_{pj}N_{qj}^*  = \delta_{pq},\qquad\qquad
  	\sum_{j=1}^n N_{jp}^*N_{jq} = \delta_{pq}.
$$
To see this, note that for all $u,v \in \mathcal{H}$,
the following identities hold true:
\begin{equation*}\begin{split}
	  &\langle u,v\rangle\,\delta_{pq} =
	  \langle u\otimes \chi_p,
	  e^{i\check F}e^{-i\check F}\,v\otimes \chi_q\rangle
	  = \langle u\otimes \chi_p,
	  e^{i\check F}Pe^{-i\check F}\,v\otimes \chi_q\rangle,\\
	  &\langle u,v\rangle\,\delta_{pq} =
	  \langle u\otimes \chi_p,
	  e^{-i\check F}e^{i\check F}\,v\otimes \chi_q\rangle
	  = \langle u\otimes \chi_p,
	  e^{-i\check F}Pe^{i\check F}\,v\otimes \chi_q\rangle
\end{split}\end{equation*}
for $p,q=1,\ldots,n$, where $P := \sum_{r=1}^n I \otimes 
\chi_r\chi_r^*$ is the orthogonal projection onto 
$\mathcal{H}\otimes\mathcal{S}$ and we have used that $\lambda_{j} 
\mathcal{S} \subset \mathcal{S}$ in the last step.  Using the definition 
of $N_{pq}$, the claim is easily established.
The next condition that we will check is
$$
  	K+K^* = - \sum_{p=1}^n L_pL_p^*.
$$
Since 
$$
	K+K^* = \sum_{i,j=1}^n G_i (W_{ij} + W_{ji}^*)G_j,\qquad
	\sum_{p=1}^n L_pL_p^* = \sum_{i,j=1}^n
	G_i \left(\sum_{p=1}^nY_i^pY_j^{p*}\right)G_j,
$$
the claim would follow immediately if
we can show that for all $i$ and $j$
$$
	  W_{ij} + W_{ji}^* = -\sum_{p=1}^nY_i^pY_j^{p*}.
$$
To see this, we begin by noting that $-f(x)f(x)^* = 2\,(\cos(x) -1)/x^2$.
Furthermore, since $\mu_j\chi_0 \in \mathcal{S}$ and $\lambda_{j} 
\mathcal{S} \subset\mathcal{S}$, we find that $f(\check F)^*\,v\otimes 
\mu_j \chi_0 \in \mathcal{H}\otimes\mathcal{S}$.  Therefore
$$
	f(\check F)f(\check F)^*\,v\otimes \mu_j \chi_0 =
	f(\check F)Pf(\check F)^*\,v\otimes \mu_j \chi_0,
$$
and we can write
\begin{equation*}
  	\left\langle u\otimes\mu_i \chi_0,
  	\frac{2\cos(\check F) - 2I}{\check F^2}\,
	v\otimes \mu_j\chi_0\right\rangle =
	  -\langle u\otimes\mu_i \chi_0,
	  f(\check F) P f(\check F)^*\,v\otimes \mu_j\chi_0\rangle.
\end{equation*}
Using the definition of $Y_i^p$, we find that for all $u,v \in 
\mathcal{H}$
\begin{equation*}
  - \sum_{p=1}^n \langle u,Y_i^pY_j^{p*}\,v\rangle =
  \left\langle u\otimes\mu_i \chi_0,
  \frac{2\cos(\check F) - 2I}{\check F^2}\,v\otimes 
	\mu_j\chi_0\right\rangle.
\end{equation*}
On the other hand, note that $g(x) + g(x)^* = 2\,(\cos(x) -1)/x^2$. 
We thus obtain
   \begin{equation*}
   \langle u,(W_{ij} + W_{ji}^*)\,v\rangle = \left\langle 
	u\otimes\mu_i \chi_0,
	  \frac{2\cos(\check F) - 2I}{\check F^2}\,v\otimes 
	\mu_j\chi_0\right\rangle
\end{equation*}
for all $u,v \in \mathcal{H}$ from the definition of $W_{ij}$.
The claim is established.

The last condition that we need to check reads
$$
  M_p = -\sum_{q= 1}^n N_{pq}L_q^*.
$$
To see this, we begin by noting that $-e^{ix}f(x)^* = f(x)$.
Since $\mu_j\chi_0 \in \mathcal{S}$
and $\lambda_{j}\mathcal{S} \subset \mathcal{S}$, we find
that $f(\check F)^*\, G_jv\otimes \mu_j\chi_0 \in \mathcal{S}$
for all $j=1,\ldots,m$.  Therefore 
$$
	f(\check F)\, G_jv\otimes \mu_j \chi_0 
	= -e^{i\check F}Pf(\check F)^*\, G_jv\otimes \mu_j \chi_0.
$$
We obtain
\begin{equation*}
  \sum_{j=1}^m\langle u\otimes\chi_p,
  f(\check F)\, G_jv\otimes \mu_j\chi_0\rangle =
  - \sum_{j=1}^m\langle u \otimes \chi_p,
  e^{i\check F}Pf(\check F)^*\, G_j v\otimes\mu_j\chi_0\rangle.
  \end{equation*}
Using the definitions of $N_{pq}$, $Y_j^q$, $L_q$ and $M_p$, we find
that the claim holds true.  \qed

\bibliographystyle{plain}
\bibliography{ref}

\end{document}